\documentclass[11pt]{article}
\usepackage{amsmath,amssymb,amsthm,amsfonts}
\usepackage{fullpage}
\usepackage{url}
\usepackage{color}
\usepackage[]{algorithmic}
\usepackage{algorithm}
\usepackage{graphicx}

\DeclareMathAlphabet{\pazocal}{OMS}{zplm}{m}{n}

\newcommand{\Oh}{\mathcal{O}}

\newcommand{\inprod}[1]{\left\langle #1 \right\rangle}

\newtheorem{theorem}{Theorem}

\newtheorem{lemma}{Lemma}

\newtheorem{definition}{Definition}

\newcommand{\proofbelow}{3pt}
\newcommand{\afterproof}{\hfill $\blacksquare$ \par \vspace{\proofbelow}}
\renewenvironment{proof}{\noindent\textbf{Proof.}\,}{\afterproof}

\author{Vasileios Nakos\thanks{Harvard University. \texttt{vasileiosnakos@g.harvard.edu}. Supported in part by NSF grant IIS-1447471.}} 

\title{One-Bit ExpanderSketch for One-Bit Compressed Sensing}

\begin{document}

\maketitle

\begin{abstract}
Is it possible to obliviously construct a set of hyperplanes $\mathcal{H}$ such that you can approximate a unit vector $x$ when you are given the side on which the vector lies with respect to every $ \textbf{h} \in \mathcal{H}$?
In the sparse recovery literature, where $x$ is approximately $k$-sparse, this problem is called one-bit compressed sensing and has received a fair amount of attention the last decade. In this paper we obtain the first scheme that achieves almost optimal measurements and sublinear decoding time for one-bit compressed sensing in the non-uniform case. For a large range of parameters, we improve the state of the art in both the number of measurements and the decoding time.

\end{abstract}

\section{Introduction}

Compressed sensing is a signal processing technique for reconstructing an approximatelly sparse signal, given access to fewer samples than what the Shannon-Nyquist Theorem requires. This technique was initiated the last decade \cite{candes2005decoding, donoho2006compressed} and has received an enormous amount of attention, because of its many applications in fields like machine learning, signal processing, computer vision, genetics etc. In this setting, one obtains $m$ linear measurements of a signal $x \in \mathbb{R}^n$:

	\[	y = Ax,	\]

where $A \in \mathbb{R}^{m \times n}$, and wants to approximatelly reconstruct $x$, exploiting prior information about its sparsity. This setting has been examined very carefully under different guarantees and many algorithms have been suggested.

However, in modern acquisition systems measurements need to be quantized: that it means that we have access only to $y= Q(Ax)$ for some $Q: \mathbb{R}^m \rightarrow \mathcal{A}^m$ \cite{boufounos20081}. In other words, $Q$ maps every element of the encoded vector to an element to a finite alphabet $\mathcal{A}$. The most common paradigm is when $\mathcal{A} = \{-1,1\}$ and

	\[	y = \mathrm{sign}(Ax),	\]
where the sign function is applied to any element of the vector. 
In hardware systems such as the analog-to-digital converter (ADC), quantization is the primary bottleneck limiting sample rates \cite{adcsurvey, le2005analog}. Moreover, as indicated in \cite{le2005analog}, the sampling rate has to decrease exponentially in order for the number of bits to be increased linearly. Furthmore, power consumption is dominated by the quantizer, leading to increased ADC costs. Thus, the one-bit compressed sensing framework provides a way to disburden the quantization bottleneck by reducing the sampling rate, i.e. the total number of measurements \cite{boufounos20081}. 

Apart from having important applications, the problem of one-bit compressed sensing is also interesting from a theoretical perspective, as it is a natural and fundamental question on high-dimensional geometry. One can think of it in the following way: can we construct a set of hyperplanes $\mathcal{H}$ such that we can approximate the direction a $k$-sparse vector $x \in \mathbb{R}^n$ given $\mathrm{sign}(\inprod{x,h})$, for all $h \in \mathcal{H}$? If we want a uniform guarantee, i.e. being able to approximate the direction of $x$ for every $x$, this means that every region defined by the hyperplanes and the sphere must have ``small' diameter. Othewise, if we want to reconstruct the direction of $x$ with some target probability, then we demand that most regions defined by the sphere and the hyperplane to have small diameter. The latter formulation is very closely related to the problem of random hyperplane tesselations \cite{plan2014dimension}.

In this work, we focus on designing a scheme for one-bit compressed sensing that enables sub-linear decoding time in the universe size $n$. Sub-linear decoding time has been extensively in the sparse recovery literature \cite{gilbert2007one,gilbert2012approximate,porat2012sublinear,gilbert2013l2,gilbert2017forall}, but overlooked in the literature of one-bit compressed sensing; the only paper that explored sub-linear decoding time is \cite{nakos2017fast}.

\subsection{Previous Work}

The problem of one-bit compressed sensing was introduced in \cite{boufounos20081}, and has received a fair amount of attention till then; one can see \cite{li2018survey} for details. Efficient algorithms, which proceed by by solving linear or convex programs when the sensing matrix consists of gaussians, appear in \cite{plan2013one,plan2013robust,gopi2013one}. Algorithms that are based on iterative hard-thresholding have been suggested in \cite{jacques2013quantized,jacques2013robust}. Moreover, the paper of Plan and Vershyin \cite{plan2014dimension} studies the very relevant problem of random hyperplane tesselations. The authors in \cite{gopi2013one,acharya2017improved} give also combinatorial algorithms for support-recovery from one-bit measurements using combinatorial structures called union-free families. Moreover, \cite{nakos2017fast} gives combinatorial algorithms for one-bit compressed sensing that run in sub-linear time, using ideas from the data streams literature and combinatorial group testing. 

The work of \cite{baraniuk2016one} introduces schemes for one-bit compressed sensing for the scenario where the underlying singal is sparse with respect to an overcomplete dictionary rather than a basis; this scenario is common in practice. Researchers have also tried to reduce the reconstruction error by employing different techniques and under different models. One approach suggested is Sigma-Delta quantization \cite{knudson2016one,gunturk2010sigma}. If adaptivity is allowed and, moreover, the measurements take the form of threshold signs, the authors in \cite{baraniuk2017exponential} show that the reconstruction error can be made exponentially small.

\subsection{Our Contribution}

In this paper, we study the non-uniform case under adversarial noise and give the first result that achieves sublinear decoding time and nearly optimal $\Oh(\delta^{-2}k + k\log n)$ measurements, where $\delta$ is the reconstruction error, $k$ is the sparsity and $n$ is the universe size. For clearness, this scheme allows reconstruction of a fixed $x\in \mathbb{R}^n$ and not of all $x \in \mathbb{R}^n$; we refer to this a non-uniform guarantee.

We compare with two previous schemes, which are the state of the art. The first scheme appears in \cite{plan2013robust}, which achieves $\delta^{-2}k \log (n/k)$ measurements and $\mathrm{poly}(n)$ decoding time, while the other appears in \cite{nakos2017fast} and achieves $\Oh( \delta^{-2} k + k \log(n/k) (\log k + \log \log n))$ measurements and $\mathrm{poly}(k,\log n)$ decoding time. We mention that the aforementioned two works are incomparable, since they exchange measurements and decoding time. However, generalizing \cite{nakos2017fast} and using the linking/clustering idea of \cite{larsen2016heavy} (which is closely related to list-recoverable codes), we are able to almost get the best of both worlds. Our scheme is strictly better the scheme of \cite{plan2013robust} when $k\leq n^{1-\gamma}$, for any constant $\gamma$; we note that the exponent of $k$ in our running time is the same as the exponent of $n$ in the running time of the relevant scheme of \cite{plan2013robust}.

We note that \cite{plan2013robust} discusses also uniform guarantees for the one-bit compressed sensing problem. Our result is non-uniform and thus incomparable with some of the results in that paper; the relevant parts from \cite{plan2013robust} are Theorem 1.1 and subsection 3.1. It is important to note that the guarantee of our algorithm \textbf{cannot} be achieved in the uniform setting, even when linear measurement are allowed \cite{cdd09} (i.e. we do not have access only to the sign of the measurement), thus a comparison is meaningful (and fair) only with a non-uniform algorithm.

\subsection{Preliminaries and Notation}

For a vector $x \in \mathbb{R}^n$ we define $H(x,k) = \{ i \in [n]: |x_i|^2 \geq \frac{1}{k} \|x_{-k}\|_2^2 \}$. If $ i\in H(x,k)$, we will say that $i$ is a $1/k$-heavy hitter of $x$. For a set $S$ we define $x_S$ to be the vector that occurs after zeroing out every $i \in [n] \setminus S$. We define $\mathrm{head}(k)$ to be the largest $k$ in magnitude coordinates of $x$, breaking ties arbitrarily, and we define $x_{-k}= x_{[n] \setminus \mathrm{head}(k)}$, which we will also refer to as the tail of $x$. Let $\mathcal{S}^{n-1} = \{x \in \mathbb{R}^n: \|x\|_2 =1 \}$. For a number $\theta$ we set $\mathrm{sign}(\theta) = 1$ if $\theta \geq 0$, and $-1$ otherwise. For a vector $v = (v_1,v_2,\ldots,v_n)$ we set $\mathrm{sign}(v) = (\mathrm{sign}(v_1), \mathrm{sign}(v_2),\ldots, \mathrm{sign}(v_n))$. We also denote $\mathcal{P}([n])$ to be the powerset of $[n]$.

\begin{definition}[Vertex Expander]
Let $\Gamma: [N] \times [D] \rightarrow [M]$ be a bipartite graph with $N$ left vertices, $M$ right vertices and left degree $D$. Then, the graph $G$ will be called a $(k,\zeta)$ vertex expander if for all sets $ S \subseteq [N], |S| \leq k$ it holds that $\Gamma(S) \geq (1 - \zeta) |S| D$.
\end{definition}

\subsection{Main Result}

The main result of our paper is the following.

\begin{theorem}
There exists a distribution $\mathcal{D} \in \mathbb{R}^{m \times n}$, a procedure $\textsc{Dec}:\{-1,+1\}^m \rightarrow \mathbb{R}^n$ and absolute constants $C_1,C_2>1$ such that

\[ \forall x \in \mathcal{S}^{n-1} : \mathbb{P}_{\Phi \sim \mathcal{D}}[ \hat{x} = \textsc{Dec}(\mathrm{sign}(\Phi x)): \|x-\hat{x}\|_2^2 > 2 \|x_{-k}\|_2^2 + \delta]  \leq e^{-C_1 \delta^{-1}k} + n^{-C_2},\]

and $\|\hat{x}\|_0 = O(k)$.

The number of rows of $\Phi$ is $ m = \Oh(k \log n + \delta^{-2}k)$, and the running time of $\mathrm{Dec}$ is $\mathrm{poly}(k,\log n)$.
\end{theorem}

It should clear that since $\hat{x}$ is $O(k)$-sparse, we do not need to output an $n$-dimensional vector, but only the positions where the vector is non-zero.

\subsection{Overview of our Approach}

The one-bit compressed sensing framework has a neat geometrical representation: one can think of every measurement $\mathrm{sign}(\inprod{\Phi_j,x})$ indicating on which side of the hyperplane $\Phi_j$ the vector $x$ lies. One of the results of \cite{plan2013robust} shows that this is possible with $\Oh(\delta^{-2}k \log (n/k))$ random hyperplanes when random post-measurement noise $v$ is added, i.e. $y = \mathrm{sign}(\Phi x + v)$; the paper gives also other, very intersesting results, but we will not focus on them in this work. To achieve sublinear decoding time we do not pick the hyperplanes (measurements) at random, but we construct a structured matrix that allows us to find all $1/k$-heavy hitters of $x$. This approach also has been followed in one of the schemes of \cite{nakos2017fast}. There the author implemented the dyadic trick \cite{cormode2008finding} in the one-bit model, showing that it is possible to recover the heavy hitters of $x$ from one-bit measurements, using $O(k \log (n/k) (\log k +\log \log n))$ measurements. Our results is an extension and generalization of that paper, along with the linking and clustering technique of \cite{larsen2016heavy}.

In the core of our scheme, lies the design of a randomized scheme which is analogous to the ``partition heavy hitters'' data structure of \cite{larsen2016heavy}; we call this scheme $\textsc{One-Bit PartitionPointQuery}$. More concretely, the question is the following: given a partition $\mathcal{P}$ of the universe $[n]$, is it possible to decide if a given set $S \in \mathcal{P}$ is heavy, when we are given access only to one-bit measurements? We answer this question in the affirmative and then combine this routine with the graph clustering technique of \cite{larsen2016heavy}.
 We thus show that, similarly to that paper, it is possible to reduce the problem of finding the heavy coordinates in the one-bit framework to the same clustering problem.

\subsection{Toolkit}

\begin{lemma}[Chernoff Bound]

Let $X_1,\ldots,X_r$ be Bernoulli random variables with $\mathbb{E}[X_i] = p$. There exists an absolute constant $c_{ch}$ such that 

\[	\mathbb{P} \left[  |\sum_i X_i - pr| > \epsilon pr \right] \leq e^{-c_{ch}\epsilon^{-2} pr} \]

\end{lemma}

\begin{lemma}[Bernstein's Inequality]
There exists an absolute constant $c_B$ such that for independent random variables $X_1,\ldots,X_r,$ with $|X_i| \leq K$ we have that

\[	\forall \lambda>0, \mathbb{P}\left[ |\sum_i X_i - \mathbb{E} \sum_i X_i | > \lambda \right] \leq e^{-C_B \lambda/\sigma^2} + e^{-C_B \lambda/K},		\]

where $\sigma^2 = \sum_i \mathbb{E}(X_i - \mathbb{E}X_i)^2$.

\end{lemma}

\begin{theorem}[Fixed Signal, Random Noise Before Quantization \cite{plan2013robust}]\label{plan}

Let $x \in \mathbb{R}^N $ and $G \in \mathbb{R}^{m \times N}$, each entry of which is a standard gaussian. If $ y = \mathrm{sign}(G x + v)$, where $v \sim \mathcal{N}(0,\sigma^2I)$, then the following program

		\[ \hat{x} = \mathrm{argmax} \inprod{y,G x},	s.t.~\|z\|_1 \leq \sqrt{k}\]

returns a vector $\hat{x}$ such that $\|x - \hat{x}\|_2^2 \leq \delta$, as long as 

	\[	m = \Omega( \delta^{-2}(\sigma^2+1)k \log(N/k) ).		\]
\end{theorem}

\section{Main Algorithm}

Our algorithm proceeds by finding a set $S$ of size $\Oh(k)$ containing all coordinates $i \in H(x,k)$ and then runs the algorithm of \cite{plan2013robust}, by restrictring on columns indexed by $S$. The scheme that is used to find the desired set $S$ is guaranteed by the following Theorem.

\begin{theorem}\label{onebitheavyhitters}
There exists a randomized construction of a matrix $\Phi \in \mathbb{R}^{m' \times n}$, a decoding procedure $\mathrm{OneBitHeavyHitters}: \{-1,1\}^{m'} \rightarrow \mathcal{P}([n])$ and an absolute constant $c$, such that $S = \mathrm{OneBitHeavyHitters}(\mathrm{sign}(\Phi x))$ satisfies the following, with probability $1-\frac{1}{n^{C_1}}$. a) $|S| \leq ck$, and b) $\forall i \in H(x,k), i \in S$. Moreover, the number of rows of $\Phi$ equals $m' = \Oh(k \log n)$ and the running time of $\mathrm{OneBitHeavyHitters}$ is $\Oh(k \cdot \mathrm{poly}(\log n))$.
\end{theorem}

Given the above theorem we show how to prove the Theorem $1$.

\begin{proof}

We vertically concatenate the matrix $\Phi$ from Theorem \ref{onebitheavyhitters} and the matrix $G$ guaranteed by Theorem \ref{plan}. Then, we run the algorithm $\mathrm{OneBitHeavyHitters}(\mathrm{sign}(\Phi x))$ to obtain a set $S$. Then we run the following algorithm:
	\[	\hat{x} = \mathrm{argmax} \inprod{y, G_S z},s.t.~ \|z\|_1 \leq \sqrt{k}.	\]

 Last, we output $\hat{x}$. Since $ Gx = G_Sx_S + G_{[n] \setminus S}x_{[n] \setminus S}$, and 
$G_{[n] \setminus S}x_{[n] \setminus S} \sim \mathcal{N}(0, \|x_{[n] \setminus S}\|_2^2 I)$ and $\|x_{[n]\setminus S}\|_2 \leq 1$, by combining the guarantees of theorems \ref{plan} and \ref{onebitheavyhitters}  we have that 

\[ \|x-\hat{x}\|_2^2 = \|x_S - \hat{x}_S \|_2^2 + \|x_{[n] \setminus S}\|_2^2 \leq \delta + 2\|x_{-k}\|_2^2,  \]

because $\|x_{[n] \setminus S}\|_2^2 \leq \|x_{-k}\|_2^2 + \sum_{i \in \mathrm{head}(k) \setminus H(x,k) } x_i^2 \leq \|x_{-k}\|_2^2 + k \frac{1}{k}\|x_{-k}\|_2^2 = 2 \|x_{-k}\|_2^2$.
\end{proof}

\textbf{Remark}: From the discussion in this subsection, it should be clear than any algorithm that runs in linear time in $n$ and has the same guarantees as as Theorem \ref{plan} immediatelly implies, by our reduction, an algorithm that achieves $\Oh(k \mathrm{poly}(\log n))$ time. Thus, any subsequent improvement of that type over \cite{plan2013robust} gives an improvement of our main result in a black-box way.

\subsection{Reduction to small Sparsity}

The following trick is also used in \cite{larsen2016heavy}. If $k = \Omega( \log n)$, we can hash every coordinate to $\Theta(k/ \log n)$ buckets and show that it suffices to find the $\frac{1}{\log n}$-heavy hitters in every bucket separately. Here, we give a proof for completeness. First, we state the following lemma, which is proven in section 2.4.

\begin{theorem} \label{thm:smallsparsity}

Let $C',C_0$ be absolute constants and suppose that $k \leq  C' \log n$. Then there exists a randomized construction of a matrix $\Phi \in \mathbb{R}^{m'' \times n}$ with $m'' = \Oh( k \log n)$ rows, such that given $y = \mathrm{sign}(\Phi x)$, we can find, with probability $1 - n^{-C_0}$ and in time $\Oh(\mathrm{poly}(\log n))$, a set $S$ of size $\Oh(k)$ containing every $i \in H(x,k)$.

\end{theorem}

Given this lemma, we show how to prove Theorem 2. This lemma is also present in \cite{larsen2016heavy}, but, for completeness, we prove it again here.

\begin{proof}
If $k < C' \log n$, we run the algorithm guaranteed by the previous lemma. Otherwise, we pick a hash function $g:[n] \rightarrow [C''k/\log n]$ and for $j \in [C''k/ \log n]$ we obtain set $S_j$ using lemma. We then output the union of all these sets. Define $z = C''k/ \log n$. We argue correctness.

For $j \in [C''k/ \log n]$ we use the Chernoff Bound to obtain that

\[ \mathbb{P} \left[  |g^{-1}(j) \cap H(x,k)| \geq \log n \right] \leq e^{-C'''\log n }.	\]

We will now invoke Bernstein's inequality for the random variables $\left\{X_i = \textbf{1}_{ g(i) = j }\right\}_{ i \in [n] \setminus H(x,k)}$; for these variables we have $K < \frac{1}{k} \|x_{-k}\|_2^2$ and

\[	\sigma^2 < \sum_{ i \in [n] \setminus H(x,k)}x_i^4 (z^{-1} - z^{-2}) \leq \frac{k}{z} \|x_{-k}\|_2^4 \sum_{i \in [n] \setminus H(x,k)} x_i^2 = \frac{k}{z} \|x_{-k}\|_2^4	\]

\[	\mathbb{P} \left[ |\sum_{i \in  g^{-1}(j) \setminus H(x,k)} x_i^2 \geq \frac{\log k}{k} \|x_{-k}\|_2^2 \right] \leq e^{-C'''\log n}.	\]

By a union-bound over all $2z = 2C''k/ \log n$ events, $2$ for every buckets $j \in [z]$, we get the proof of the lemma.

\end{proof}

Our paper now is devoted to proving Theorem \ref{thm:smallsparsity}.

\subsection{\textsc{One-BitPartitionPointQuery}}

In this section we prove the following Theorem, which is the main building block of our algorithm.


\begin{theorem} \label{partitionsketch}
Let $x \in \mathbb{R}^n$ and a partition $\mathcal{P} = \{P_1,P_2,\ldots, P_T\}$ of $[n]$. There exists an oblivious randomized construction of a matrix $Z \in \mathbb{R}^{m \times n}$ along with a procedure $\textsc{One-BitPartitionPointQuery}: [T] \rightarrow \{0,1\}$, where $m = \Oh( k \log (1/\delta) )$, such that given $y= \mathrm{sign}(Zx)$ the following holds for $j^* \in [T]$. 
\begin{enumerate}
\item If $P_{j^*}$ contains a coordinate $i \in H(x,k)$, then $\textsc{One-BitPartitionPointQuery}(j^*)=1$ with probability $1-\delta$.
\item If there exist at least $ck$ indices such that $\|x_{P_j}\|_2 \geq \|x_{P_{j^*}}\|_2$, then $\textsc{One-BitPartitionPointQuery}(j^*)=0$ with probability $1-\delta$.

\end{enumerate}
 Moreover, The running time of is $\Oh(\log (1/\delta))$.

\end{theorem}

We describe the construction of the the matrix $Z$. We are going to describe the matrix as a set of linear measurements on the vector $x$. For $i \in [n], j \in [T], B \in [C_Bk], \ell \in [3], r \in [C_3 \log (1/\delta)]$ we pick the following random variables:
\begin{enumerate}
\item fully independent hash functions $h_{r,\ell}: [T] \rightarrow [C_Bk]$.
\item random signs $\sigma_{j,B,\ell,r}$. Intuitively, one can think of this random variable as the sign assigned to set $P_j$ in bucket $B$ of sub-iteration $\ell$ of iteration $r$. 
\item normal random variables $g_{i,r}$. One can think of this random variable as the gaussian associated with $i$ in iteration $r$.
\end{enumerate}

Then, for every $B \in [C_Bk], \ell \in [3], r \in [C_B \log (1/\delta)]$  we perform linear measurements

\[	z_{B,\ell,r} = \sum_{j \in h_{r,\ell}^{-1}(B) } \sigma_{j,B,\ell,r} \sum_{ i \in P_j} g_{i,r} x_i,		\]
as well as measurements $-z_{B,\ell,r}$ (the reason why we need this will become clear later).

Of course we have access only to the sign of the measurement: $y_{B,\ell,r} = \mathrm{sign}(z_{B,\ell,r})$. We slightly abuse notation here, as $y$ is described as a $3$-dimensional vector; it is straightforward to see how this vector can be mapped to a $1$-dimensional vector.

We will make use of the following lemmata. The value $C_B$ is a large enough constant, chosen in order for the analysis to work out. Before proceeding with the lemmas, we pick constants $C_u,C_d$ such that
\begin{enumerate}
\item $\mathbb{P}_{Y \sim \mathcal{N}(0,1)}[|Y| < C_u] = \frac{19}{20}$.
\item $\mathbb{P}_{Y \sim \mathcal{N}(0,1)}[|Y| > C_d] = \frac{19}{20}$.
\end{enumerate}

\begin{lemma}\label{positives}
Fix $i^* \in H(x,k)$, $j^*$ such that $i^* \in P_{j^*}$, as well as $r \in [C_3 \log (1/\delta)]$. We also set $B_\ell = h_{r,\ell}(j^*)$. Then, with probability at least $\frac{3}{5}$  we have that for all $ \ell \in [3]$ either

\[	y_{B_\ell,\ell,r} = \sigma_{j^*,B_\ell,\ell,r} ~\mathrm{or}~y_{B_\ell,\ell,r} = -\sigma_{j^*,B_\ell,\ell,r}.	\]

\end{lemma}

\begin{proof}


For the need of the proof we define $\mathcal{B}^{-1}_\ell = h_{r,\ell}^{-1}(h_{r,\ell}(j^*))$. First, observe that for all $\ell \in [3]$ that the random variable

\[ Y_\ell = \sum_{j \in \mathcal{B}^{-1}_\ell \setminus \{j^*\} } \sigma_{j,B_\ell,\ell,r} \sum_{ i \in P_j} g_{i,r} x_i \]

is distributed as \[ \sqrt{ \left(\sum_{j \in \mathcal{B}^{-1}_l \setminus \{j^*\}}\sum_{i \in P_j} x_i^2 \right) } \cdot \mathcal{N}(0,1). \]

Observe that with probability at least $\frac{19}{20}$, $|Y_\ell|$ will be at most 
\[ C_{u}\sqrt{ \sum_{j \in \mathcal{B}^{-1}_\ell \setminus \{j^*\}}\sum_{i \in P_j} x_i^2 }. \]

Define  \[ Z_\ell = \sum_{j \in \mathcal{B}^{-1}_\ell \setminus \{j^*\}}\sum_{i \in P_j}x_i^2.\] 

Consider now the set $\mathcal{P}_{\mathrm{bad}}$ of $P_j$, $j \in [T] \setminus\{j^*\}$ for which there exists $i \in H(x,k)$ such that $i \in P_j$. Since there are at most $2k$ elements in $\mathcal{P}_{\mathrm{bad}}$, with probability at least $1-\frac{2}{C_B}$ it holds that $\mathcal{B}^{-1}_{\ell} \cap \mathcal{P}_{\mathrm{bad}} = \emptyset$. Let this event be $\mathcal{W}$. It is a standard calculation that $\mathbb{E}[Z_l|\mathcal{W}] \leq \frac{1}{C_Bk}\|x_{-k}\|_2^2$. Invoking Markov's inequality one gets that $Z_l$ is at most $\frac{20}{C_B k}\|x_{-k}\|_2^2$ with probability at least $\frac{19}{20}$.
Putting everything together, this gives that 
	\[|Y_\ell| > C_{u} \sqrt{\frac{20}{C_B k}} \|x_{-k}\|_2\]
 with probability $\frac{1}{20}$. The probability that there exist $l \in [3]$ such that $|Y_\ell| > C_{u} \sqrt{\frac{20}{C_B k}} \|x_{-k}\|_2$ is at most $\frac{3}{20}$. We now observe that the

\[ |\sum_{i \in P_{j^*}} g_{i,r} x_i| \geq C_d \|x_{P_{j^*}}\|_2\geq C_d\frac{1}{\sqrt{k}} \|x_{-k}\|_2\] with probability at least $\frac{19}{20}$. The above discussion implies that with probability at least $\frac{15}{20}$ the quantity $|\sum_{i \in P_{j^*}} g_{i,r} x_i|$ is larger than $|Y_\ell|$, for all $l \in [3]$, if $C_d/ \sqrt{k} > C_{u} \sqrt{\frac{20}{C_B k}}$. This means that, with probability at least $\frac{3}{4}$, the sign of $z_{B_\ell,\ell,r}$ will be determined by the sign of $\sigma_{j^*,B_\ell,\ell,r} \sum_{ i \in P_{j^*}} g_{i,r} x_i$ for all $\ell\in[3]$. This implies that if $\sum_{i \in P_{j^*}}g_{i,r} x_i > 0$, we will get that $y_{B_\ell,\ell,r} = \sigma_{j^*,B_\ell,\ell,r}$. On the other hand, if $\sum_{i \in P_{j^*}} g_{i,r} x_i< 0$ then $y_{B_\ell,\ell,r} = -\sigma_{j^*,B_\ell,\ell,r}$. This gives the proof of the lemma.

\end{proof}

\begin{lemma}\label{negatives}
Let $j^*$ such that $ \|x_{P_{j^*}}\|_2 > 0$. We also define $B_\ell =h_{r,\ell}(j^*)$. Assume that there exist at least $ck$ indices $j$ such that $\|x_{P_{j}}\|_2 \geq \|x_{P_{j^*}}\|_2$, for some absolute constant $c$. Then, with probability $\frac{3}{5}$, there exists indices $\ell_1,\ell_2 \in [3]$ such that

\[	y_{B_{\ell_1},\ell_1,r} = \sigma_{j^*,B_{\ell_1},\ell_1,r} ~\mathrm{and}~y_{B_{\ell_2},\ell_2,r} = -\sigma_{j^*,B_{\ell_2},\ell_2,r}.	\]

\end{lemma}
\begin{proof}
For the need of the proof we also define $\mathcal{B}^{-1}_\ell = h^{-1}_{r,\ell}(h_{r,\ell}(j^*))$. Fix $\ell \in [3]$.
Let $\mathcal{P}_{\mathrm{good}}$ be the set of indices $j \in [T]$ such that $\|x_{P_j}\|_2 \geq \|x_{P_{j^*}}\|_2$. Let the random variable $Z_\ell$ be defined as 
\[	Z_\ell = |\{ j \in \mathcal{P}_{\mathrm{good}}  \setminus \{j^*\}: j \in \mathcal{B}^{-1}_\ell \}|.	\]

Observe now that $\mathbb{E}[Z_\ell] = \frac{ck}{C_Bk} = \frac{c}{C_B}$ and moreover $Z_l$ is a sum of independent Bernoulli random variables with mean $\frac{1}{C_Bk}$, hence a standard concetration bound gives that, for $c$ large enough, $Z_\ell$ will be larger than $4C_d^2 C_u^2$ with probability $\frac{19}{20}$. This implies that 

\[	\sum_{j \in \mathcal{B}_\ell^{-1}\setminus\{j^*\}} \|x_{P_j}\|_2^2 \geq 4C_d^2 C_u^2 \|x_{P_{j^*}}\|_2^2.		\]

for all $\ell \in [3]$. This implies that, for any $\lambda \in \mathbb{R}$,

\[	\mathbb{P} \left[ |\sum_{j \in \mathcal{B}_\ell^{-1}\setminus \{j^*\} } \sigma_{j,B_\ell} \sum_{ i \in P_j} g_{i,r} x_i| \geq \lambda \right] \geq\]
\[ \mathbb{P} \left[ 2C_d C_u \|x_{P_{j^*}}\|_2 \cdot |\mathcal{N}(0,1)| \geq \lambda  \right]. \]

The above implies that 
\[\mathbb{P} \left[ |\sum_{j \in B_l \setminus \{j^*\} } \sigma_{j,B_\ell,\ell,r} \sum_{ i \in P_j} g_{i,r} x_i| \geq 2C_u \|x_{P_{j^*}}\|_2 \right] \geq \frac{19}{20} \]

and moreover 

\[	\mathbb{P} \left[ |\sum_{i \in \mathcal{P}_{j^*}} g_{i,r} x_i | \leq C_u \|x_{P_{j^*}}\|_2 \right] \geq \frac{19}{20},\]

which implies that with probability $\frac{17}{20}$ we have that

\[	|\sum_{j \in \mathcal{B}^{-1}_\ell \setminus \{j^*\} } \sigma_{j,\mathcal{B}_\ell^{-1},l,r} \sum_{ i \in P_j} g_{i,r} x_i|\leq 2|\sum_{i \in \mathcal{P}_{j^*}} g_{i,r} x_i |.		\]

Observe now that $y_{B_\ell,\ell,r}$ is the same as the sign of

$\sum_{j \in B_\ell \setminus \{j^*\} } \sigma_{j,\mathcal{B}_\ell^{-1},l,r} \sum_{ i \in P_j} g_{i,r} x_i$, which, because of the random signs, means that 

\[	\mathbb{P} \left[ y_{B_\ell,\ell,r} = 1\right]	= \frac{1}{2}. \] 

Moreover, we get that $y_{B_\ell,\ell,r}$ and $\sigma_{j^*,B_\ell,\ell,r}$ are independent. Conditioned on the previous events, the probability that either 

\[ y_{B_\ell,\ell,r}= \sigma_{j^*,B_\ell,\ell,r} \]	
 
for all $\ell \in [3]$, or 

\[ y_{B_\ell,\ell,r}= -\sigma_{j^*,B_\ell,\ell,r} \]	

for all $\ell\in[3]$, is $\frac{2}{8}$. This gives the proof of the claim since $\frac{3}{20} + \frac{2}{8} \leq \frac{8}{20} = \frac{2}{5}.$ 

\end{proof}

We are now ready to proceed with the proof of Theorem \ref{partitionsketch}.

\begin{proof}

We iterate over all $r \in [C_3 \log(1/\delta)]$ and count the number of ``good'' repetitions: a repetition $r$ is good if for all $\ell\in[3]$, $y_{h_{r,\ell}(j^*),\ell,r} = \sigma_{j,h_{r,\ell}(j^*),\ell,r} $ or $y_{h_{r,\ell}(j^*),\ell,r} = -\sigma_{j,h_{r,\ell}(j^*),\ell,r}$. We also check if there exists $l\in[3]$ such that $y_{h{r,\ell}(j^*),\ell,r} = 0$ by checking the values of $y_{h{r,\ell}(j^*),\ell,r} = 0$ and $-y_{h{r,\ell}(j^*),\ell,r} = 0$. If there exists no such $\ell$ and the number of good repetitions is at least $\lceil \frac{1}{2} C_3 \log(T/\delta) \rceil + 1$ we output $1$, otherwise we output $0$. \newline 
We proceed with the analysis. First of all, if there exists an $\ell\in[3]$ that satisfies $y_{h{r,\ell}(j^*),\ell,r} = 0$, this would mean that $\|x_{P_{j^*}}\| = 0$. Let us assume that this is not the case, otherwise we can ignore $j^*$. If $i^* \in H(x,k)$ belongs to $P_{j^*}$, for some $j^*$, using Lemma 
\ref{positives} the expected number of good iterations equals 
$(3/5)C_3 \log(1/\delta)$,and by a Chernoff Bound we get that at least 
$(2.6/3) \cdot(3/5)C_3 \log(1/\delta) = (2.6/5) C_3 \log(1/\delta) $ repetitions will be good with probability 
\[ 1 - e^{-\Omega( \log( |T|/\delta))} \geq 1- \delta,\] for large enough $C_3$. In the same way, using Lemma \ref{negatives} we can bound by $\delta$ the probability that a set $P_{j^*}$, for which there exist at least $ck$ set $P_j$ with $\|x_{P_j}\|_2 \geq \|x_{P_{j^*}}\|_2$, has more than $\lceil \frac{1}{2} C_3 \log(T/\delta) \rceil - 1$ good repetitions. This concludes the proof of the lema.
 
\end{proof}

The following lemma is immediate by taking $\delta = T^{-C_0-1}$ and taking a union-bound
over all $j \in [T]$.

\begin{lemma}[\textsc{One-BitPartitionCountSketch}]\label{partitioncountsketch}
Let $x \in \mathbb{R}^n$ and a partition $\mathcal{P} = \{P_1,P_2,\ldots, P_T\}$ of $[n]$. There exists a randomized construction of a matrix $Z \in \mathbb{R}^{m \times n}$, such that given $y=\mathrm{sign}(Zx)$, we can find in time $\Oh(k \log T)$ a set $S$ of size $\Oh(k)$ that satisfies contains every $j \in [T]$ for which there exists $i \in H_k(x) \cap P_j$. Moreover, the failure probability is $T^{-C_0}$.
\end{lemma}

\subsection{\textsc{One-Bit} $b$-tree}

We now describe the scheme of \textsc{One-Bit} b-tree. The $b$-tree is a folkore data structure in streaming algorithms, first appearing in \cite{cormode2008finding} in the case of vectors with positive coordinates. The version of the $b$-tree we are using here is more closely related in \cite{larsen2016heavy}. We remind the reader that the aforementioned papers treated the case where we have access to $\Phi x$ and not only to $\mathrm{sign}(\Phi x)$. 
Here,  we describe it a sensing matrix associated with a decoding procedure, rather than a data structure. Given the $b$-tree, we can find elements $i \in H(x,k)$ and get an analog of Theorem $1$; however, this would only give $1/\mathrm{poly}(\log n)$ failure probability. Getting $1/\mathrm{poly}(n)$ failure probability requires using the \textsc{ExpanderSketch} algorithm of \cite{larsen2016heavy}In fact, we can use the \textsc{One-Bit} $b$-tree to speed up the \textsc{One-Bit} \textsc{ExpanderSketch} decoding procedure, but since our overall scheme  already has a polynomial dependence on $k$ in the running time due to the application of Theorem, this will not give us any crucial improvement. However, we believe that it might of independent interest in the sparse recovery community.

The following lemma holds.

\begin{lemma}\label{btree}

Let $k,b<n$ be integers. There exists a randomized construction of a matrix $A \in \mathbb{R}^{M \times n}$ such that given $y = \mathrm{sign}( A x)$ we can find a set $S$ of size $\Oh(k)$ such that $\forall i \in H(x,k), i \in S$. The total number of measurements equals

\[		M = \Oh( k \frac{\log(n/k)}{\log b} (\log(k/\delta) + \log\log(n/k) - \log \log b ))	\]

the decoding time is

\[ 	\Oh( b k \frac{\log(n/k)}{\log b} (\log(k/\delta) + \log\log(n/k) - \log \log b ))	\]

and the failure probability is $\delta$.

\end{lemma}

\begin{proof}

Let $R$ be the smallest integer such that $kb^R \geq n$; this means that $R = \left \lceil \log(n/k)/\log b \right \rceil$. For $r=0,\ldots,R$ we use the \textsc{One-Bit} \textsc{PartitionCountSketch} scheme guaranteed by Lemma, with $\delta = \delta/(bkR)$ and partition $\mathcal{P}_r = \{ \{1,\ldots,\left \lceil \frac{n}{kb^r} \right \rceil\}, \{\left \lceil \frac{n}{kb^r} \right \rceil +1, \ldots,2\left \lceil \frac{n}{kb^2} \right \rceil\} , \ldots  \}$, of size $T_r = \Theta(kb^r)$. 

The total number of measurements equals

\[	\Oh( R k \log(bkR/\delta)) = \Oh( k \frac{\log(n/k)}{\log b} (\log(k/\delta) + \log\log(n/k) - \log \log b )).		\]

We can think of the partitions $T_1,T_2,\ldots, T_R$ as the levels of a $b$-ary tree; for every set if $T \in T_r$, there are $b$ sets $T' \in T_{r+1}$ which are neighbours of $T$.
The decoding algorithms starts at quering the \textsc{One-Bit} \textsc{PartitionCountSketch} for $r=0$ to obtain a set $S_0$. Then, for every $i \in [1,r]$, it computes all the neighbours of $S_{r-1}$, where the  for a total of $\Oh(b |S|)$ sets. Then using \textsc{One-BitPartitionPointQuery} we query every new partition, to obtain a set $S_r$ of size $\Oh(k)$. The output of the algorithm is the set $S_{R}$. The running time then is computed as 

\[	\Oh( b R k \log(bkR/\delta) =\Oh( b k \frac{\log(n/k)}{\log b} (\log(k/\delta) + \log\log(n/k) - \log \log b ))\]
\end{proof}

From the above lemma, we get the following result, by carefully instatianting the parameter $b$.

\begin{lemma}
There exists a $b$ such that the \textsc{One-Bit} $b$-tree uses $\Oh(\gamma^{-1}k \log (n/\delta))$ measurements and runs in time $\Oh(\gamma^{-1}(k \log (n/\delta))^{2+\gamma})$, for any arbitarily constant $\gamma$.
\end{lemma}

\begin{proof}

We set $ b= (k \log (n/\delta))^{\gamma}$ and observe that the number measurements is at most

			\[ \Oh \left( \ k \frac{\log n}{\gamma(\log (k/\delta) + \log \log n)} (\log (k/\delta) + \log \log n) \right) = \Oh\left ( \frac{1}{\gamma} k \log (n/\delta) \right),		\]

while the decoding time becomes

			\[ \Oh \left( (k \log (n/\delta))^{\gamma} k \frac{\log n}{ \gamma(\log (k/\delta)+ \log \log n)}(\log (k/ \delta) + \log \log n) \right) = \Oh \left( \frac{1}{\gamma} (k \log (n/\delta))^{2+\gamma} \right)				\]
\end{proof}

\subsection{\textsc{One-Bit} \textsc{ExpanderSketch}}

In this subsection we prove Theorem \ref{thm:smallsparsity}. Given the results about $\textsc{One-BitPartitionCountSketch}$ we developed in the previous sections, the proof of the theorem is almost identical to \cite{larsen2016heavy} with a very simple modification. For completeness, we go again over their construction. We remind the reader that in our case $k = \Oh( \log n)$. \newline 

\textbf{Construction of the Sensing Matrix}: We first pick a code $\mathrm{enc}:\{0,1\}^{\log n} \rightarrow \{0,1\}^{\Oh(\log n)}$, which corrects a constant fraction of errors with linear-time decoding; such a code is guaranteed by \cite{spielman1996linear}. We then partition $\mathrm{enc}(i)$ into $s= \Theta( \log n/ \log \log n)$ continuous substrings of length $t = \Theta(\log \log n)$. We denote by $\mathrm{enc}(i)_j$ the $j$-th bitstring of length $t$ in $\mathrm{enc}(i)$.

We define $s$ hash functions  $h_1,h_2,\ldots,h_s :[n] \rightarrow [\mathrm{poly}(\log n)]$. Let also $F$ be an arbitrary $d$-regular connected expander on the vertex set $[s]$ for some $d= \Oh(1)$. For $j\in[s]$, we define $\Gamma_j \subset [s]$ as the set of neighbours of $j$. Then, for every $j \in [n]$  we define the bit-strings

\[	m_{i,j} = h_j(i) \circ \mathrm{enc}(i)_j \circ h_{\Gamma_1(j)}(i) \ldots \circ h_{\Gamma_d(j)}(i),	\]

and the following partitions $\mathcal{P}^{(j)}$ containg set $P^{(j)}_{m_{i,j}}$, where $m_{i,j}$ is a string of $\Theta(t)$ bits, such that:

		\[ \forall i \in [n],  i \in P^{(j)}_{m_{i,j}}	\]

Then for every partition $\mathcal{P}^{(j)}$ we pick a random matrix $\Phi^{(j)}$ using Lemma \ref{partitioncountsketch} with sparsity $k$, as well as  a random matrix $Z^{(j)}$ using Lemma \ref{partitionsketch} with sparsity $k$ and failure probability $\frac{1}{\mathrm{poly}(\log n)}$. Each of these matrices has $\Oh( k \log ( 2^{\Oh(t)})) = \Oh(k t) = \Oh( k \log \log n)$ rows. The total number of rows is $\Oh(s k \log \log n) = \Oh(k \log n)$.
Then our sensing matrix is the vertical concatenation of $\Phi^{(1)},Z^{(1)},\ldots,\Phi^{(s)},Z^{(s)}$.

\textbf{Decoding Algorithm}: For every $j \in [s]$ we run the decoding algorithm of Lemma \ref{partitioncountsketch} on matrix $\Phi^{(j)}$ to obtain a list $L_j$ of size $\Oh(k)$ such that every ``heavy'' set of $\mathcal{P}^{(j)}$ is included. The running time in total is $ m \cdot k \cdot \mathrm{poly}(\log n) = \mathrm{poly}(\log n)$. For every $j \in [s]$, we now have that:

\begin{itemize}
\item With probability $1/\mathrm{poly}(\log n)$, $h_j$ perfectly hashes every $P_{m_{i,j}}^{(j)}$ for every $i \in H(x,k)$. 
\item With probability $1/\mathrm{poly}(\log n)$, for every $i^* \in H(x,k)$, $\|x_{P^{(j)}_{m_{i,j}}}\|_2 \geq \frac{9}{10} \|x_{-k}\|_2.$ 
\item With probability $1/\mathrm{poly}(\log n)$, the decoding procedure on $\Phi^{(j)}$ succeeds. This follows by taking a union bound over the events of the previous two bullets and the failure probability guarantee of Lemma \ref{partitioncountsketch} in our instance.
 \end{itemize}

We call by ``name'' of $P^{(j)}_{m_{i,j}}$ the $\Oh(\log \log n)$-length substring of  bits of $m_{i,j}$, which correspond to the bits of $h_j(i)$.
 We then filter out vertices in layer $j$, by keeping only those that have unique names. Our next step is to point-query every set $z \in L_j$ using the matrices $Z^{(j)}$ and Theorem \ref{partitionsketch} and keep the largest $\Oh(k)$ coordinates; this is the difference with \cite{larsen2016heavy}, since we can implement only one-bit point query. Now we let $G$ be the graph created by including the at most $(d/2)\sum_{j=1}^s L_j$ edges suggested by the z's across all $L_j$, where we only include an edge if both endpoints suggest it. Now the algorithm and analysis proceeds exactly as \cite{larsen2016heavy}.








\bibliographystyle{alpha}
\bibliography{biblio}

\end{document}